\documentclass[12pt]{amsart}
\usepackage{hyperref, amsmath,amssymb,amsfonts,amsbsy,mathtools,cite,setspace}
\usepackage{graphicx}
\usepackage{amssymb}
\usepackage{amsthm}
\usepackage{enumerate}
\usepackage{wasysym}
\usepackage{cite}
\usepackage{tikz}
\usepackage{float}
\usetikzlibrary {graphs} 
\usetikzlibrary{shapes.geometric, backgrounds, calc}
\usepackage{mathtools}

\usepackage{braket}
\usepackage{color}

\def\ket#1{| #1 \rangle}
\def\bra#1{\langle #1 |}
\def\kb#1#2{|#1\rangle\!\langle #2 |}
\def\bk#1#2{\langle #1 |#2\rangle}

\def\be{\begin{eqnarray}}
\def\ee{\end{eqnarray}}
\def\bee{\begin{eqnarray*}}
\def\eee{\end{eqnarray*}}

\newtheorem{defn}{Definition}

\newtheorem{thm}{Theorem}
\newtheorem{exa}{Example}
\newtheorem{lem}{Lemma}

\newtheorem{rmk}{Remark}
\newtheorem{cor}{Corollary}

\newcommand{\C}{{\mathbb C}}

\newcommand{\M}{{\mathbb M}}
\renewcommand{\H}{{\mathcal H}}

\newcommand{\operp}{$\bigcirc$\kern-.91em{$\perp$}}

\def\rk{{\rm rk}}

\def\be{\begin{eqnarray}}
\def\ee{\end{eqnarray}}
\def\bee{\begin{eqnarray*}}
\def\eee{\end{eqnarray*}}
\def\ot{\otimes}

\begin{document}

\title[Chordal Graphs and Distinguishability of Quantum States]{Chordal Graphs and Distinguishability of Quantum Product States}
\author[D.W.Kribs, C.Mintah, M.Nathanson, R.Pereira]{David W. Kribs$^{1}$, Comfort Mintah$^{2}$, Michael Nathanson$^3$, Rajesh Pereira$^{1}$}

\address{$^1$Department of Mathematics \& Statistics, University of Guelph, Guelph, ON, Canada N1G 2W1}
%\address{$^2$African Institute for Mathematical Sciences, Muizenberg, Cape Town, 7945, South Africa}
\address{$^2$Centre for Education in Mathematics and Computing, University of Waterloo, Waterloo, ON, Canada N2L 3G1}
\address{$^3$Department of Mathematics, Harvard University, Cambridge, MA, USA 02138}

\begin{abstract}
We investigate a graph-theoretic approach to the problem of distinguishing quantum product states in the fundamental quantum communication framework called local operations and classical communication (LOCC). We identify chordality as the key graph structure that drives distinguishability in one-way LOCC, and we derive a one-way LOCC characterization for chordal graphs that establishes a connection with the theory of matrix completions. We also derive minimality conditions on graph parameters that allow for the determination of indistinguishability of states. We present a number of applications and examples built on these results. 
\end{abstract}



\maketitle

\section{Introduction}

One of the most important quantum communication frameworks is referred to as local (quantum) operations and classical communication, and is denoted LOCC. It includes many basic quantum communication protocols such as quantum teleportation and data hiding \cite{Teleportation, terhal2001hiding,eggeling2002hiding}. The subclass of LOCC protocols called one-way LOCC, wherein the communicating parties perform their measurements in a prescribed order, still encapsulates many of the important cases and has been studied from the early days of LOCC investigations      \cite{Walgate-2000,Nathanson-2005,fan2004distinguishability,N13,cosentino2013small,yu2012four}. 

A central focal point of the subject, of relevance to many LOCC applications, is the development of techniques to distinguish amongst known sets of quantum states using only one-way LOCC operations. Our work on LOCC quantum state distinguishability began with the application and development of techniques from matrix and operator theory \cite{kribs2017operator,kribsquantum2019,lattice2019}. Most importantly for this work, we recently introduced an orthogonal graph representation approach to the problem of distinguishing sets of quantum product states  \cite{kribs2020vector,kribs2021operator}. 

In this paper, we deepen this investigation with a number of results and applications that further strengthen the connections between graph theory, LOCC, and linear algebra and operator theory. Specifically, we begin by formulating a notion of a distinguishability set and prove its equivalence to certain decompositions of positive semidefinite matrices. We identify chordality as the main graph structure that drives distinguishability, and we establish a characterization of chordal graphs in terms of one-way LOCC distinguishability of bipartite product states. Our key technical device in linking these two notions comes from matrix completion theory. We then turn our focus to decidedly non-chordal graphs, and derive results that allow for determination of indistinguishability of states. We identify minimality conditions in terms of standard graph parameters imply the indistinguishability of sets of states. 
Throughout the work we also present a number of illustrative examples and applications of our main results, including further new results in special cases as consequences, improvements and new proofs for some previous results, and new graph-theoretic perspectives for some seminal examples in the subject. 

This paper is organized as follows. In the next section we include the necessary background material from quantum information and graph theory. Section~3 contains the main results on distinguishability and chordal graphs, and Section~4 flips the focus to indistinguishability and other relevant graph structures. Applications and examples for each part are included as subsections.

\section{Preliminaries}

We use standard quantum information notation, including the Dirac notation of ket's $\ket{\psi}$ for pure states (unit vectors) in a given Hilbert space and bra's $\bra{\psi}$ for the dual vectors. The standard orthonormal basis for $\mathbb{C}^d$ is denoted by $\{\ket{0}, \ket{1}, \ldots , \ket{d-1}\}$, and the corresponding basis for $\mathbb{C}^d \otimes \mathbb{C}^n$ is written in shorthand as $\ket{ij} = \ket{i} \otimes \ket{j}$.  

Our basic setting is a situation in which two parties, called Alice and Bob, each in control of a quantum system represented on finite-dimensional Hilbert spaces $\mathcal H_A$ and $\mathcal H_B$, are physically separated  so that their measurement protocols are restricted to those only using local quantum operations and classical communications (LOCC). Their joint system $\mathcal H_A \otimes \mathcal H_B$ has been prepared in a pure state from a known set of states $\mathcal S = \{ \ket{\psi_i}  \}$, and Alice and Bob would like to determine the value of $i$ only using LOCC. We consider the sub-paradigm called one-way LOCC, wherein Alice can perform a measurement on her system and then communicate (classically) the result to Bob, who can then perform a measurement on his system.  
Mathematically, a one-way LOCC measurement is of the form $\M = \{ A_k \otimes B_{k,j} \}$, with the positive operators making up the measurement outcomes satisfying $\sum_k A_k = I_A$ and $\sum_{j} B_{k,j} = I_B$ for each $k$. If the outcome $A_k \otimes B_{k,j}$ is obtained for any $k$, the conclusion is that the prepared state was $\ket{\psi_j}$. 
More generally, a positive operator valued measure (POVM) is a set of operators $E_k\geq 0$ such that $\sum_k E_k = I$. 

As in \cite{kribs2020vector}, we shall focus on the problem of distinguishing sets of quantum product states, which are tensor products of states on the individual subsystems, $\ket{\psi} = \ket{\psi^A} \otimes \ket{\psi^B}$. It is this class of quantum states that allows for a link with graph theory. See \cite{lovasz1989orthogonal,lovasz2000correction,fallat2007minimum,booth2008minimum,fallat2011variants,barioli2010zero,booth2011minimum,barioli2011minimum,hackney2009linearly,rose1970triangulated}  as entrance points into the literature on graphs and their uses in linear algebra. Here we present the basic notions that will be used throughout the paper, with additional definitions introduced as needed in subsequent sections. 

Let $G = (V,E)$ be a {\it (simple) graph} with vertex set $V$ and edge set $E$. For $v,w\in V$, we write $v \sim w$ if the edge $( v,w ) \in E$.  The {\it complement} of $G$ is the graph $\overline{G} = (V, \overline{E})$, where the edge set $\overline{E}$ consists of all two-element sets from $V$ that are not in $E$. Another graph $G'$ is a {\it subgraph} of $G$, written $G' \leq G$, if $V' \subseteq V$ and $E' \subseteq E$ with $v,w\in V'$ whenever $( v,w ) \in E'$. For a subset of the vertices $V' \subset V$, $G' = (V', E')$ is the {\it induced subgraph} of $G$ on $V'$ when $E' = \{ (v,w) \in E: v,w \in V'\}$. The {\it complete graph} on $n$-vertices, is $K_n = (\{v_1,\dots , v_n    \}, E)$ such that $E = \{ (v,w) : v\neq w \in V \}$.  The {\it cycle graph}  on $n$-vertices is $C_n = (\{v_1,\dots , v_n    \}, E)$ such that $E = \{ (v_i, v_{i+1}) : 1 \leq i \leq n-1 \} \cup \{ (v_n, v_1) \}$. 

A set of graphs $\{ G_i = (V_i, E_i) \}$ {\it covers} a graph $G= (V,E)$ if $V = \cup_i V_i$ and $E = \cup_i E_i$.
A collection of graphs $\{G_i \}$ is a {\it clique cover} for $G$ if $\{G_i \}$ covers $G$ and if each of the $G_i$ is a `clique'; i.e., a complete graph.
The {\it clique cover number} $\mathrm{cc}(G)$ is the smallest possible number of subgraphs contained in a clique cover of $G$. 
Note that a clique cover can be thought of as a collection of (not necessarily disjoint) induced subgraphs of $G$, each of which is a complete graph. It is a cover if every edge is contained in at least one of the cliques. A {\it simplicial vertex} in $G$ is a vertex whose neighbor vertices form a clique; i.e., any two neighbors of a simplicial vertex are adjacent in the graph.  

The following class of graphs is central to our analysis. 

\begin{defn}
A graph $G$ is \emph{chordal} if it possesses no chordless cycles of length four or greater; i.e., all cycles in $G$ of four or more vertices have a cycle chord. Equivalently, every induced cycle of $G$ must have exactly three vertices.  
\end{defn}

The link we make use of between graph theory and Hilbert space structure is the following. 
Given a graph $G= (V,E)$, a function $\phi: V \rightarrow \mathbb{C}^d\backslash \{0\} $ is an {\it orthogonal representation} of $G$ if for all vertices $v_i \ne v_j\in V$,
$v_i \not\sim v_j$ if and only if $\bk{\phi(v_i)}{\phi(v_j)} = 0$. 
For more details on orthogonal representations of graphs we point the reader to \cite{fallat2007minimum,lovasz1989orthogonal,booth2008minimum,booth2011minimum}. 

As introduced in \cite{kribs2020vector}, every set of product states naturally defines orthogonal graph representations as follows. 

\begin{defn}
Let $\mathcal S = \{ \ket{\psi^A_k}\otimes \ket{\psi^B_k} \}_{k=1}^r$ be a set of product states on $\mathcal H_A \otimes \mathcal H_B$. The graph of these states from Alice's perspective is the unique graph $G_A$ with vertex set $V = \{ 1,2, \ldots, r\}$ such that the map $k \mapsto  \ket{\psi^A_k}$ is an orthogonal representation of $G_A$. Likewise, the graph of the states from Bob's perspective is the graph $G_B$ with vertex set $V$ such that $k \mapsto \ket{\psi^B_k}$ is an orthogonal representation of $G_B$.
\end{defn}

Observe that a set of product states are mutually orthogonal precisely when Alice's graph is a subgraph of the complement of Bob's graph; i.e., $G_A \leq \overline{G_B}$.

\section{Distinguishability Sets and Chordal Graphs}\label{section3}

We begin by recalling a pair of notions from matrix theory. 
Let $M$ be a positive definite matrix, then the {\it support} of $M$ is the set  $\{ i: m_{ii}\neq 0\}$.  
Let $\mathcal S=\{ \ket{\phi_i} \}_{i=1}^n$ be a collection of pure states all on a Hilbert space $\mathcal H$, then the {\it Gram matrix} of $\mathcal S$ is the $n$ by $n$ (positive semidefinite) matrix whose $(i,j)$th entry is the inner product $\bk{\phi_i}{\phi_j}$. 

\begin{defn}  \label{def:DistinguishabilitySet}
Let $\mathcal S=\{ \ket{\phi_i} \}_{i=1}^n$ be a collection of pure states on the Hilbert space $\mathcal H$, and let $\{A_k\}_{k=1}^m$ be a collection of subsets of $\mathcal S$.  Then $\{A_k\}_{k=1}^m$ is said to be a {\em distinguishability set} for $\mathcal S$ if there exists a POVM $\{ E_k \}_{k=1}^m$ on $\mathcal H$ with $m$ outcomes with the property that if we perform the measurement on $\ket{\phi}\in \mathcal S$ and we can get outcome number $k$, then  $\ket{\phi}$ must be in $A_k$; i.e., $\ket{\phi_i}\notin A_k$ implies that $E_k \ket{\phi_i}=0$.  
\end{defn}

As an illustrative special case, if the states $\ket{\phi_i}$ are mutually orthogonal and the sets $A_k$ are disjoint with union equal to the set $\{ 1 ,\ldots , n\}$, then $E_k = \sum_{j\in A_k} \kb{\phi_j}{\phi_j}$ defines such a POVM. 
In general, the distinguishability sets of $\mathcal S$ and the Gram matrix of $\mathcal S$ can be related in the following result.

\begin{thm} \label{firstmain} 
Let $\mathcal S=\{ \ket{\phi_i} \}_{i=1}^n$ be a collection of pure states on the Hilbert space $\mathcal{H}$ and let $M$ be the Gram matrix of $\mathcal S$.  Let $\{ T_k\}_{k=1}^m$ be a collection of subsets of $\{ 1,2,...,n\}$ and let $A_k =\{ \ket{\phi_i} \}_{i\in T_k}$ for $1\le j\le m$.  Then the following are equivalent: 
\begin{enumerate}
\item $\{A_k\}_{k=1}^m$ is a distinguishability set for $S$.
\item There exists a set of non-zero $n$ by $n$ positive semidefinite matrices $\{M_k\}_{k=1}^m$ such that $\sum_{k=1}^m M_k=M$ and $M_k$ is supported on $T_k$ for $1\le k\le m$.
    \end{enumerate}
\end{thm}

%\begin{proof} 
%(1) implies (2):  Let $\{E_k\}_{k=1}^{m}$ be the POVM which implements this distinguishability set.  Now for  $1\le k\le m$, let $M_k$ be the $n$ by $n$ matrix whose $(i,j)$th entry is $\bra{\phi_i} E_k \ket{\phi_j}$.  Then $\sum_{k=1}^m M_k=M$ as $\sum_{k=1}^m E_k = I_{\mathcal H}$, and for $1\le k\le m$,  $M_k$ is supported on $T_k$ as $i \notin T_k$ implies $0 = \mathrm{tr}(E_k \kb{\phi_i}{\phi_i}) = \bra{\phi_i}E_k\ket{\phi_i}$. 

%(2) implies (1):  Without loss of generality let us assume that the states $\{ \ket{\phi_i} \}_{i=1}^n$ span $\mathcal{H}$ and that all $M_k$ are rank one.  Let $X$ be a linear operator from $\mathbb{C}^n$ to  $\mathcal{H}$ which maps the standard basis vector $e_j$ to $\ket{\phi_j}$ for all $j$.  Then $M=X^*X$.  We now can use $\sum_{k=1}^m M_k=M$ to show that $M=Y^*Y$ where $Y$ is an $m$ by $n$ matrix whose $(i,j)$ entry is zero if $j\not \in T_i$.                {\bf DK: give more detail on this direction of the proof?}   
%\end{proof}

\begin{proof} 
Without loss of generality let us assume that the states $\{ \ket{\phi_i} \}_{i=1}^n$ span $\mathcal{H}$ (so $d := \dim\H \le n$). Let $X$ be the rank-$d$ linear operator from $\mathbb{C}^n$ to  $\mathcal{H}$ which maps the standard basis vector $\ket{j}$ to $\ket{\phi_j}$ for all $j$. That is, $X = \sum_{j = 1}^n \kb{\phi_j}{j}$. It is immediate to see that $M=X^*X$. 

To prove that (1) implies (2):  Let $\{E_k\}_{k=1}^{m}$ be the POVM that implements this distinguishability set $\{A_k\}$.  For each $1\le k\le m$, we can define the positive semidefinite matrix $M_k = X^*E_kX$. That is, $M_k$ is the $n$ by $n$ matrix whose $(i,j)$th entry is $\bra{\phi_i} E_k \ket{\phi_j}$. We can check that $\sum_k M_k = \sum_k X^*E_kX = X^*X = M$, since the elements of the POVM sum to the identity. 

By Definition \ref{def:DistinguishabilitySet}, $\bra{\phi_i} E_k \ket{\phi_j} \ne 0$ implies that $\ket{\phi_i}$ and $\ket{\phi_j}$ are both in $A_k$. Hence $M_k$ is supported on $T_k$.

For the converse (2) implies (1):  We are given the decomposition $M = \sum_{k = 1}^m M_k$, with each $M_k$ supported on $T_k$, and we wish to construct a  POVM to distinguish the sets $\{A_k\}$, with each $A_k = \{ \ket{\phi_i} \}_{i\in T_k}$. We can write each $M_k = \sum_i v_{k,i}v_{k,i}^*$ for  unnormalized vectors $v_{k,i} \in \C^n$. Since $0 \le M_k \le M$, we see that each $v_{k,i}$ is in the image of $M = X^* X$, and hence in the image of $X^*$. Thus, there exists a state $\ket{\varphi_{k,i}}$ and a positive constant $\mu_{k,i}$ such that $X^* \ket{\varphi_{k,i}} = \frac{1}{\mu_{k,i}} v_{k,i}$, for each $k$ and $i$.

We can now explicitly define $E_k = \sum_i \mu_{k,i}^2 \kb{\varphi_{k,i}}{\varphi_{k,i}}$. We claim that $\{E_k\}$ is a POVM that implements $\{A_k\}$ as a distinguishability set. 

Firstly, each $E_k\ge 0$ by construction, so showing it is a POVM requires only that $\sum_k E_k = I_d$. We calculate: 
\begin{equation}
X^*\left(\sum_k E_k\right)X = \sum_k X^*E_kX = \sum_k M_k = M = X^*X . 
\end{equation}
Hence, $X^*\left(I_{d} - \sum_k E_k\right)X = 0$, which implies that $\sum_k E_k = I_{d}$ (since the rank of $X$ is the dimension $d$ of $\H$ and the nullity of $X^*$ is 0). 

Lastly, we show that this POVM distinguishes $\{A_k\}$. Given our definition of $X$, we see that for each $i$ and $k$, 
\begin{equation}
\bra{i} M_k \ket{i} =\bra{i} X^*E_k X \ket{i} = \bra{\phi_i}E_k\ket{\phi_i}. 
\end{equation}
Since $\ket{\phi_i} \notin A_k$ implies that $\bra{i} M_k \ket{i}= 0$, it also implies $\bra{\phi_i}E_k\ket{\phi_i} = 0$, and so our POVM distinguishes the $\{A_k\}$. 
\end{proof}

For one-way LOCC on product states, we can define a variant of the Gram matrix that uses only the component of the state that we are going to measure first. 

\begin{defn} 
Let $\mathcal S=\{ \ket{\psi_i^A} \otimes  \ket{\psi_i^B}\}$ be a collection of product states on the Hilbert space $\mathcal H= \mathcal H_A \otimes \mathcal H_B$. Then the {\em Alice Gram matrix} of $\mathcal S$ is the $n$ by $n$ matrix whose $(i,j)$th entry is  $\bk{\phi_i^A}{\phi_j^A}$, and the {\em Bob Gram matrix} of $\mathcal S$ is the $n$ by $n$ matrix whose $(i,j)$th entry is  $\bk{\phi_i^B}{\phi_j^B}$. 
\end{defn}

%\subsection{Main Result}

We shall follow the notation from \cite{paulsen1989schur} in the following.  If $G=(V,E)$ is a graph, then let $$\mathcal S_G=\{[a_{ij}]: a_{ij}=0 \,\, \mathrm{if} \,\, (i,j)\notin E \}$$ be the operator system defined by the graph, where the matrix sizes are given by the cardinality of the vertex set (and non-zero diagonal entries are allowed). Observe that $G_1\leq G_2$ implies $S_{G_1} \subseteq S_{G_2}$. We note that the notation $\mathcal S_G$ has more restrictive meaning in the minimum rank literature; but this will not be significant in what follows. 

\begin{cor}\label{gram} 
Let $\mathcal S$ be a collection of orthonormal product states with Alice Gram matrix $M$. Then $\mathcal S$ is distinguishable with one way LOCC with Alice going first if and only if $M$ is the sum of rank one positive semidefinite matrices in $\mathcal S_{\overline{G_B}}$.
\end{cor}

\begin{proof} 
Let $M$ be the Alice Gram matrix of $\mathcal S$ and suppose we have  rank one positive semidefinite elements $M_k$ of $\mathcal S_{\overline{G_B}}$ for all $k$, such that $M=\sum_{k=1}^m M_k$. Then there exists a set of (possibly overlapping) cliques $\{T_k\}_{k=1}^m$ that define a graph $G\subseteq \bigcup_{k=1}^m T_k$ with $G\leq \overline{G_B}$, such that $M_k$ is supported on $T_k$. (The graph $T_k$ is determined by the support of $M_k$, and the fact it is a clique follows since $M_k$ is rank one.)  By Theorem~\ref{firstmain}, there is a POVM on $\mathcal H_A$ corresponding to this clique cover of $G$. If Alice applies this POVM (which acts on her component of the states) and gets say outcome $k$, then she can restrict the state to be determined to those corresponding to elements of $T_k$, which is a clique in $G$ and hence in $\overline{G_B}$. As the corresponding vertices are mutually disconnected in $G_B$, Bob is left with an orthonormal set of states that he can then distinguish. 

Conversely,  if $\mathcal S$ is distinguishable by one-way LOCC with Alice going first, then there must be a clique cover $\{T_k\}_{k=1}^m$ of $\overline{G_B}=G$ which corresponds to a distinguishability set for the Alice components of the set of states.  By Theorem~\ref{firstmain}, this means  that there exists positive semidefinite matrices $\{M_k\}_{k=1}^m$ such that $\sum_{k=1}^m M_k=M$ and $M_k$ is supported on $T_k$ for $1\le k\le m$.  Each $M_k$ can be written as a sum of rank-one positive semidefinite matrices (the total number of which equals the rank of $M_k$). The diagonal entries of any of these rank one matrices must be less than or equal to the corresponding diagonal entries of $M_k$. Hence these rank one matrices must necessarily be supported on $T_k$, which means they are all members of $ \mathcal S_{\overline{G_B}}$.
%But $T_k$ is a clique and hence chordal, so every $M_k$ can be written a sum of rank one positive semidefinite matrices in $\mathcal S_{T_k}\subseteq \mathcal S_{\overline{G_B}}$.
\end{proof}

We can use this result to derive an LOCC-based characterization of graph chordality. In doing so, we link this LOCC condition with other matrix theoretic conditions via the main result from \cite{paulsen1989schur}. We state the relevant part of \cite[Theorem 2.4]{paulsen1989schur} for our purposes as the following lemma. 

\begin{lem}\label{jfaresult}
Let $G$ be a graph.  Then every positive semidefinite matrix in $\mathcal S_G$ is the sum of rank one positive semidefinite matrices in $\mathcal S_G$ if and only if $G$ is a chordal graph. 
\end{lem}

We can now prove our main result. 

\begin{thm} \label{mainthm}
Let $G$ be a graph with $n$ vertices.  Then $G$ is chordal if and only if every collection of $n$ product states having Alice graph $G_A$ and Bob graph $G_B$ with $G_A \leq G \leq \overline{G_B}$ is distinguishable with one-way LOCC with Alice going first.    
\end{thm}

\begin{proof} 
Suppose $G$ is chordal and we have a collection of $n$ product states with $G_A \leq G \leq \overline{G_B}$.  Let $M$ be the Gram Matrix of the Alice components of the product states.  Then $M\in \mathcal S_{G}$.  Since $G$ is chordal, Lemma~\ref{jfaresult} implies that $M$ can be written as a sum of rank one positive semidefinite matrices in $\mathcal S_G\subseteq \mathcal S_{\overline{G_B}}$.  Hence $\mathcal S$ is distinguishable by one-way LOCC with Alice going first by Corollary~\ref{gram}.

Now suppose $G$ is not chordal.  Let $M\in \mathcal S_{G}$ be a positive semidefinite matrix in $S_G$ which is not the sum of rank one positive semidefinite matrices in $\mathcal S_G$, as is given by Lemma~\ref{jfaresult}.  We can assume without loss of generality that $M$ 
%is a correlation matrix (i.e., has all diagonal entries equal to one).  
has all its non-zero diagonal entries equal to 1; indeed, any general $M$ is similar to such a matrix via a positive diagonal similarity matrix, which will not disrupt the matrix belonging to $S_G$ or it not being a sum of rank one matrices in $S_G$. 
We now construct a set of product states as follows.  The Alice components of the states are chosen so as to have Gram matrix $M$ (a matrix $X$ can be constructed using the entries of $M$ with $X^*X = M$, similar to the proof of Theorem~\ref{firstmain}).  The Bob components of the states are chosen so that $\overline{G_B}=G$. By Corollary~\ref{gram}, this set of states is not distinguishable by one-way LOCC with Alice going first.
\end{proof}

%Suppose we are given a set of Alice's states in $\C^d$ with corresponding chordal graph $G= G_A$. Since the $n$ by $n$ identity matrix is a positive semidefinite matrix in $\mathcal S_G$, it can be written as a sum of rank one projections in $\mathcal S_G$. Each of these projections has support on a clique of $G$, so it can be the first step in a one-way LOCC protocol. Thus, we can always distinguish product states when Alice's graph is chordal or when there exists a chordal graph $G$ between $G_A$ and $\overline{G_B}$.

\begin{rmk}
{\rm 
This result has a number of consequences for distinguishability, some of which we will explore below. We note this result subsumes Proposition~3 of \cite{kribs2020vector}, which proved the forward direction of the result for the subclass of chordal graphs called {\it $k$-trees}. We also note that a weaker converse of the result is also relatively simple to show: Let $G_A= \overline{G_B}$ be any graph for which the minimum vector rank is strictly less than the clique cover number $\mathrm{cc}(G_A)$. Then there exist representations of $G_A$ that cannot be distinguished with one-way LOCC.  This follows immediately from looking at representations in the minimum dimension.
}
\end{rmk}

\subsection{Applications and Examples}

The results above show how one-way distinguishability is linked with chordal graph structure. Here we present some examples that illustrate this connection and its limitations. As applications we also derive a new distinguishability result, an improvement of a previous result, and a new proof for another. 

\begin{exa}
{\rm 
Consider the (unnormalized) set of qubit-qutrit product states $\ket{\psi_{i}} \in \mathbb{C}^{2} \otimes \mathbb{C}^{3} $, for $1\leq i \leq 4$, defined as follows:
\begin{align*}
\ket{\psi_{1}}& = \ket{0} \otimes \ket{0}
%\\
& \ket{\psi_{2}} &=  \left(   \ket{0} + \ket{1}  \right)  \otimes \ket{1} \\
\ket{\psi_{3}} &=  \left(   \ket{0} - \ket{1}  \right)  \otimes \ket{1} 
%\\
& \ket{\psi_{4}} &= \left(   \ket{0} +  \ket{1}  \right)  \otimes \ket{2} 
\end{align*}
Here we have Alice graph as the four cycle $C_4$ with the edge $(3,4)$ removed and the edge $(1,4)$ added. The Bob graph is the graph with four vertices and a single edge $(2,3)$. Thus, $G_A \lneq \overline{G_B}$ and both of these graphs are chordal. 

Observe these states are one-way distinguishable with Alice going first. Indeed, she can measure in her $\{\ket{+},\ket{-}\}$ basis, then communicate the result to Bob, which will tell him the state is either from the set $\{1,2,4\}$ (in the case of a $+$ result), or the set $\{1,3\}$ (in the case of a $-$ result). Bob then can measure in his $\{\ket{0},\ket{1},\ket{2}\}$ basis to determine which state they have. 
}
\end{exa}

\begin{exa}\label{Ex:4Cycle in 2D}
{\rm 
As a variant on the previous example, consider the two-qubit states $\ket{\psi_{i}} \in \mathbb{C}^{2} \otimes \mathbb{C}^{2} $, for $1\leq i \leq 4$, defined as follows:
\begin{align*}
\ket{\psi_{1}}& = \ket{0} \otimes \ket{0}
%\\
& \ket{\psi_{2}} &=  \left(   \ket{0} + \ket{1}  \right)  \otimes \ket{1} \\
\ket{\psi_{3}} &=  \left(   \ket{0} - \ket{1}  \right)  \otimes \ket{1} 
%\\
& \ket{\psi_{4}} &= \ket{1}   \otimes \ket{0} 
\end{align*}
These states cannot be distinguished via one-way LOCC with Alice going first, as Bob can only distinguish amongst the first pair or the second pair of states, and Alice cannot reduce the problem to those pairs no matter what her measurement is. 

Note in this case that the Alice graph is just the four cycle $C_4$, which is not chordal. The Bob graph is the graph with four vertices and two edges $\{(2,3), (1,4)\}$. So we have  $G_A = \overline{G_B}$, and this graph is not chordal. 
}
\end{exa}

These examples exhibit the role of chordal graphs in one-way distinguishability as given in the theorem above, and they may be viewed as special cases of the following consequence of that result. 

\begin{cor}\label{chordalcor}
Let $\mathcal S$ be a collection of orthonormal product states. If $G_A$ or $\overline{G_B}$ is a chordal graph, then the states in $\mathcal S$ are one-way distinguishable with Alice going first. 
\end{cor}

\begin{proof}
This is a straightforward consequence of the forward direction of Theorem~\ref{mainthm}, applied with $G$ equal to a chordal choice amongst the graph set $\{ G_A , \overline{G_B}\}$. 
\end{proof}

This result improves on Proposition~2 of \cite{kribs2020vector} (which just proved the result for the Bob complement graph). The converse of this result does not hold, as shown by the following example; and indeed this owes to the fact that graph chordality as characterized in the theorem depends on all possible sets of product states (and not just one) with Alice-Bob graphs sandwiched around the graph to be distinguishable. 

%\begin{exa}
%{\rm 
%Consider the four states in Example~1, with the fourth state replaced by $\ket{\psi_4} = \ket{1} \otimes \ket{2}$. These states are one-way distinguishable with Alice going first, as she can measure in her $\{\ket{+},\ket{-}\}$ basis, communicate the result to Bob, which will tell him the state is either from the set $\{1,2,4\}$ (in the case of a $+$ result), or the set $\{1,3,4\}$ (in the case of a $-$ result). Bob then can measure in his $\{\ket{0},\ket{1},\ket{2}\}$ basis to determine which state they have. The Alice graph in this case is $C_4$ and Bob graph is the graph with four vertices and a single edge $(2,3)$.

%Note in this case that the Alice graph is just the four cycle $C_4$, which is not chordal. The Bob graph is the graph with four vertices and two edges $\{(2,3), (1,4)\}$, and so we still have  $G_A = \overline{G_B}$. 
%}
%\end{exa}

\begin{exa}
{\rm 
Consider the (unnormalized) set of states $\ket{\psi_{i}} \in \mathbb{C}^{4} \otimes \mathbb{C}^{2} $, for $1\leq i \leq 4$, defined as follows:
\begin{align*}
\ket{\psi_{1}}& = \left(   \ket{0} + \ket{1}  \right) \otimes \ket{0}
%\\
& \ket{\psi_{2}} &=  \left(   \ket{1} + \ket{2}  \right)  \otimes \ket{1} \\
\ket{\psi_{3}} &=  \left(   \ket{0} + \ket{3}  \right)  \otimes \ket{1} 
%\\
& \ket{\psi_{4}} &= \left(   \ket{2} +  \ket{3}  \right)  \otimes \ket{0} 
\end{align*}
As in Example \ref{Ex:4Cycle in 2D}, the Alice graph is the four cycle $C_4$, and the Bob graph is its complement with four vertices and edges $\{(1,4),(2,3)\}$. So $G_A = \overline{G_B}$ and this graph is not chordal, but nevertheless the states are distinguishable. Indeed, Alice can measure in her $\{\ket{0},\ket{1},\ket{2},\ket{3}\}$ basis, the result of which will (respectively) tell Bob the state is amongst the pairs $\{1,3\}, \{1,2\}, \{2,4\}, \{3,4\}$, and he can then determine the state from the given pairing by measuring in his $\{\ket{0},\ket{1}\}$ basis.  

So the converse of the corollary does not hold. Also note that there is no contradiction with the theorem statement here, as, for instance, Example~2 gives a set of states with the same (non-chordal) graph that are not distinguishable. 
}
\end{exa}

The realization of the LOCC connection with chordal graphs also allows us to prove the following result, which covers a broad class of product state sets.

\begin{cor}\label{cc2cor}
Let $n,d>1$.  Suppose we have a set of orthonormal product states in $\mathbb{C}^{n} \otimes \mathbb{C}^{d} $ such that Alice's graph or Bob's complement  graph have clique cover number less than or equal to two; that is, 
\[
\min \{ \mathrm{cc}(G_A), \mathrm{cc}(\overline{G_B}) \} \leq 2 . 
\]

Then the states are distinguishable via one-way LOCC with Alice going first.
\end{cor}

\begin{proof}
Note that a graph $G$ with clique cover number less than or equal to 2 is necessarily chordal. Indeed, if it was not chordal, then we could find a cycle within the graph with four (or more) vertices such that the graph includes no edges between the cycle vertices other than on the cycle. Hence, given a clique cover of $G$, any three consecutive edges on the cycle must belong to different cliques, which means the clique cover number for the graph must be at least three.  Thus, we have $G_A$ or $\overline{G_B}$ is a chordal graph, and so the result follows from Theorem~\ref{mainthm} applied to a chordal graph from the set $\{ G_A , \overline{G_B}\}$. 
\end{proof}

For illustrative purposes, let us explicitly describe how a one-way LOCC measurement can be constructed in the special case of the previous result, where the clique cover number of Bob's complement graph is either one or two. 

\begin{exa}
{\rm 
Let $\{ \ket{\psi_k^A}\otimes \ket{\psi_k^B} \}_{k\in I}$ be a set of orthonormal product states in $\mathbb{C}^{n} \otimes \mathbb{C}^{d} $ and let $\overline{G}_{B}$ be a the complement of the corresponding Bob graph.   If $\mathrm{cc}(\overline{G}_{B})=1$ then $\{ \ket{\psi_k^B} \}_{k\in I}$ is an orthogonal set and Bob can distinguish the states without Alice performing any measurements.  So suppose $\mathrm{cc}(\overline{G}_{B})=2$, corresponding to cliques $V_1$ and $V_2$. We can then partition our index set $I$ into disjoint subsets $I=I_1\cup I_2\cup I_3$, where $I_1 = V_1 \cap \overline{V_2}$, $I_2 = V_2 \cap \overline{V_1}$ and $I_3 = V_1 \cap V_2$ (possibly empty). Both $\{ \ket{\psi_k^B} \}_{k\in I_1\cup I_3}$ and $\{ \ket{\psi_k^B} \}_{k\in I_2\cup I_3}$ are orthogonal subsets of $\mathbb{C}^{d} $ and with $\bk{\psi_j^B}{\psi_k^B}\neq 0$ if $j\in I_1$ and $k\in I_2$. Since the product states are orthogonal, $\bk{\psi_j^A}{\psi_k^A}= 0$ if $j\in I_1$ and $k\in I_2$.  Therefore $S_1= \mathrm{span} \{\ket{\psi_k^A}: k\in I_1 \}$ and $S_2=\mathrm{span} \{\ket{\psi_k^A}: k\in I_2 \}$ are two mutually orthogonal subspaces of $\mathbb{C}^{n}$.  Alice can thus perform corresponding projection-valued measurement which will leave Bob with an orthogonal subset of either $\{ \ket{\psi_k^B} \}_{k\in I_1\cup I_3}$ and $\{ \ket{\psi_k^B} \}_{k\in I_2\cup I_3}$ or $\{ \ket{\psi_k^B} \}_{k\in I_3}$, which he can then distinguish.
}
\end{exa}

We finish by providing a new proof of a result from \cite{kribs2021operator} (which in turn built on a result from \cite{kribs2020vector}) that identified a case, that of the `single qubit sender', in which distinguishability is completely described by graph conditions. 

\begin{cor}\label{singlequbit}
A set of orthonormal product states in $\mathbb{C}^{2} \otimes  \mathbb{C}^{d}$, for $d \geq 2$, is distinguishable via one-way LOCC with Alice going first if and only if there is some graph between the two graphs $G_A$ and $\overline{G_B}$ with clique cover number at most two; that is, there is a graph $G$ such that
\begin{equation}\label{singlequbitgraphcond}
G_A \leq G \leq \overline{G_B} \quad \mathit{and} \quad \mathrm{cc}\,(G) \leq 2.
\end{equation}
\end{cor}

\begin{proof}
As noted in Theorem~7 of \cite{kribs2021operator}, if the states are one-way distinguishable with Alice going first, then the existence of such a $G$ and a clique cover that satisfies $\mathrm{cc}(G)\leq 2$ follows from Theorem~1 of \cite{kribs2020vector} restricted to the case $\dim \mathcal H_A =2$.   

For the converse direction, a direct constructive proof based on the different possible clique cover sizes was presented in  \cite{kribs2021operator}. Here we simply note that $\mathrm{cc}(G)\leq 2$ implies that $G$ is chordal (as in the previous result proof), and hence this direction follows from the corresponding implication of Theorem~\ref{mainthm}. 
\end{proof}

Note that the distinguishability (or not) of Examples 1 and 2 can also be seen as a direct consequence of the graph condition of this result.

\section{LOCC and Graph Representations of Minimum Rank}\label{LOCC and graph reps}

Our main theorem builds an association between decompositions of matrices in $\mathcal{S}_G$ and POVMs that initiate an LOCC discrimination protocol, and it establishes that the chordality of $G$ is equivalent to the existence of a distinguishable set of cliques for any representation of $G$. As a contrast, here we will consider graphs that are as non-chordal as possible and investigate non-distinguishability in more detail. 

In \cite{kribs2021operator}, we built an explicit POVM to distinguish chordal graphs using an elimination ordering. This suggests that perhaps the ``least chordal'' graphs will be those with no simplicial vertices. 
This includes cycles. Any vector representation of a cyclic graph $C_n$ must have dimension at least $(n-2)$. Our first result in this section relates minimum-rank representations and the distinguishability of the cycle's edges. Note that $(n-2)$ is the minimum vector rank of the graph $C_n$, meaning that the Hilbert space dimension is as small as possible for an orthogonal representation of $G$.

\begin{lem}\label{cycleLemma}
Let $G = C_n$ be a cyclic graph with $n\ge 4$, and let $\mathcal S=\{ \ket{\phi_i} \}_{i=1}^n \subset \H$ be an orthogonal representation of $G$ with $\dim \H = n-2$. 

If $\{A_k\}$ corresponds to the collection of edges of $G$, then $\{A_k\}$ is not a distinguishability set and, hence, the Gram matrix $M$ cannot be decomposed as a sum of matrices $\{M_k\}$ with support on $\{A_k\}$.
\end{lem}

\begin{proof}
We first note that every induced subgraph of $C_n$ with $(n-1)$ vertices is the path $P_{n-1}$, and any orthogonal representation of $P_{n-1}$ has the property that any $(n-2)$ of the states are linearly independent. This follows from Corollary 3.9 of \cite{hackney2009linearly}, which states that any $|T|-1$ vertices in a tree $T$ form an OS-set, which implies that the corresponding vertex states must be linearly independent. (It was also stated in \cite{barioli2010zero} that any one vertex in a tree forms a positive semidefinite zero-forcing set, implying that the complement must form an OS-set.)

For completeness, we include a proof that vectors representing any $(|T|-1)$ vertices of a tree $T$ must be linearly independent. Suppose we have a minimal linear dependence among the $\{\phi(v)\}$, $\sum_i k_i \phi(v_i)=0$, and let $H$ be the subgraph of the tree induced by the $v_i$ for which the corresponding $k_i$ is nonzero. Since the states from different connected components of $H$ would be mutually orthogonal to one another, $H$ must be connected by minimality. Now let $u$ be a vertex adjacent to $H$ but not in $H$. Then there is a unique $w$ in $H$ adjacent to $u$. So let $\mathcal V$ be the orthogonal complement of $\phi(u)$.  Then $\phi(w)$ is not in $\mathcal V$ but $\phi(v_i)$ is in $\mathcal V$ for all the other vertex states $v_i$ in $H$, which means $\phi(w)$ cannot be part of the linear dependence, a contradiction.

Returning to our cycle $C_n$: Suppose $A_1 = \{v_1,v_2\}$ corresponds to the edge $(v_1,v_2)$. By the above argument, the $(n-2)$ states corresponding to the remaining vertices are linearly independent and hence span $\H$. If a POVM could distinguish $A_1$, then $E_1\ket{\phi_i} = 0$ for all $i >2$. Since these states span $\H$, this  implies that $E_1 = 0$. The same applies to every operator that makes up the POVM, and so no such POVM exists. Hence the edge sets are not distinguishable. 

The fact the $M$ cannot be decomposed follows from Theorem~\ref{firstmain}.
\end{proof}

This suggests a generalization. To make it easier to state, we set the following definitions.
\begin{defn}
For a graph $G$, let $$\mathcal{S}_G^+ := \{ M \in \mathcal{S}_G: M \ge 0, \,\, \Delta(M) \mbox{ is invertible}\},$$ where $\Delta$ maps the matrix $M$ to its diagonal, and define
$$\eta_+(G) := \min\{ \rk(M): M \in \mathcal{S}_G^+\}.$$
\end{defn}

In other words, if $n= |G|$, $\mathcal{S}_G^+$ is the set of $n \times n$ positive semidefinite matrices in $\mathcal{S}_G$ whose support is all of $[n]$. Every element of $\mathcal{S}_G^+$ can be written as $X^*X$, where the columns of $X$ define an orthogonal representation of a subgraph of $G$ that includes all the vertices of $G$.

Because $\mathcal{S}_G^+$ no longer contains the zero matrix, we can consider the minimum rank $\eta_+(G)$ as a possibly meaningful graph parameter. Our minimum rank parameter is similar to the Haemers number of the graph $G$\cite{haemers1978upper}, originally introduced as an upper bound of the Shannon capacity $\Theta(G)$ and equivalent to the minimum rank over matrices in $\mathcal{S}_G$ with no zeroes on the diagonal. Following the notation of  \cite{hogben2013handbook}, the Haemers number of $G$ is written $\eta(G)$, so here we adopt the notation $\eta_+(G)$ as the minimum rank when restricted to positive semidefinite matrices. 

This allows us to state our results in more graph theoretical terms. Recall that $\alpha(G)$ is the size of the largest {\it independent set} in $G$; that is, subsets of vertices of $G$ such that no two of which are connected by edges of $G$ (so independent sets in $G$ form cliques in the complement of $G$). If $G$ has an independent set of size $\alpha = \alpha(G)$, then every $M \in {\mathcal S}_G^+$ can be written with an $\alpha \times \alpha$ diagonal matrix as its upper left block, perhaps after relabelling the vertices, which implies that $\rk (M) \ge \alpha$. On the flip side, every legal coloring of $\overline{G}$ of size $k$ corresponds to a rank $k$ element of ${\mathcal S}_G^+$ (an explicit construction is given in \cite{haemers1978upper}); so the graph coloring (chromatic) number $\chi(\overline{G})$ is an upper bound for $\eta_+(G)$.

This gives the following lemma, where the first inequality follows from the definition of $\Theta(G)$, the second is a result from \cite{haemers1978upper}, the third is trivial, and the fourth follows from the discussion above. 

\begin{lem}\label{lemmainequality}
For any graph $G$, \begin{equation*} \alpha(G) \le \Theta(G) \le \eta(G) \le \eta_+(G) \le \chi(\overline{G}) . 
\end{equation*}
\end{lem}

We use this notation in the following. Note that Lemma \ref{cycleLemma} addresses this situation in the case that $G$ is a cycle. 

\begin{thm}\label{NoSimplicial} 
Let $\mathcal{S}$ be a collection of orthonormal product states in $\H_A \ot \H_B$, with associated graphs $G_A$ and $G_B$, and suppose that $G = \overline{G_B}$ has no simplicial vertices.

If $\dim \H_A = \eta_+(G)$, then the product states of $\mathcal{S}$ cannot be distinguished via one-way LOCC with Alice going first.
\end{thm}

%We have shown this is true if $G$ is a cycle. We also note that if $G$ has no simplicial vertex, then the $mr_+(G) < cc(G)$, so the states cannot be distinguished by projecting onto an orthonormal basis. [See proof below]

%So the question is, if we have a POVM that distinguishes our states, can we use it to reduce the dimension of our representation? 

%We danced around these issues in the earlier paper, it feels like there should be some answers there. 

\begin{proof}
    We prove by contradiction. Suppose there exists a POVM with operator elements $\{a_k \kb{\varphi_k}{\varphi_k}\}_k$ operating on Alice's states $\{\ket{\psi_j^A} \}_j$  that forms the first step in a one-way LOCC protocol to distinguish the states of $\mathcal{S}$. (The operators can be assumed to be rank one without loss of generality.) We construct a new graph $G'$ by adding an additional vertex $w_1$ to $G$ corresponding to $\ket{\varphi_1}$, with $w_1 \sim v_j$ if and only if $\bk{\varphi_1}{\psi_j^A} \ne 0$. Because $G_B = \overline{G} \geq G_A$, it follows that $w_1$ must be a simplicial vertex in order to initiate a one-way LOCC discrimination; i.e., Bob must be able to distinguish the states corresponding to the neighbor vertices of $w_1$, as Alice gains no information on them. 
    
    Since no vertex in $G$ is simplicial, every vertex has a neighbor outside the clique neighborhood of $w_1$, and hence each state has a component in the orthogonal complement of $\ket{\varphi_1}$.  We can perform an orthogonal vertex removal of $w_1$ \cite{hackney2009linearly} by projecting onto the orthogonal complement of $\ket{\varphi_1}$, getting a nonzero output for each $\ket{\psi_j^A}$. This creates a new set of Alice states $\{\ket{\phi_j^A} \}$ such that $\bk{\phi_i^A}{\phi_j^A} = \bk{\psi_i^A}{\psi_j^A}$ unless both $v_i$ and $v_j$ are neighbors of $w_1$ and thus neighbors of each other. This means that for any vertices in $G$ such that $v_i \not\sim v_j$, $\bk{\phi_i^A}{\phi_j^A} = \bk{\psi_i^A}{\psi_j^A} = 0$.

    We can now define the new Alice Gram matrix $M^\prime$ whose $(i,j)$ entry is $\bk{\phi_i^A}{\phi_j^A}$. It is clear that $M^\prime \in {\mathcal S}_G^+$ and that $\rk(M') < \dim \H_A$. This contradicts our assumption that $\rk(M)$ was minimal and thus proves the result. 
\end{proof}

Note that we never used the entire POVM, just the fact that there exists an operator that selects one clique with the property that none of its vertices are simplicial. If we extend this reasoning to the entire POVM, we see that a minimum rank representation cannot be distinguished with one-way LOCC with Alice going first unless $G$ contains a set of $\eta_+(G)$ simplicial vertices that are pairwise non-adjacent.  We use this idea and the fact that $\eta_+(G_A) \le \eta_+(\overline{G_B}) \le \chi(G_B)$ in the following theorem, which is our main result in this section.

\begin{thm}\label{ConverseTheorem} Let $\mathcal{S}$ be a collection of orthonormal product states in $\H_A \ot \H_B$, with associated graphs $G_A$ and $G_B$, such that $\dim\H_A = \chi(G_B)$. 

If the elements of $\mathcal{S}$ are distinguishable via one-way LOCC with Alice going first, then: 
\begin{enumerate}
    \item Alice's measurement must consist of projections onto an orthogonal basis of $\H_A$.
    \item{}\label{WeaklyPerfect} $ \alpha(\overline{G_B}) = \chi(G_B) \le \alpha(G_A)$, and $G_B$ is a weakly perfect graph.
    \item{} Alice's measurement is unique. 
    \item{} \label{ChordalCondition}There exists a chordal graph $G$ with $G_A \le G \le \overline{G_B}$.
    \item{} Every set of product states represented by the graphs $G_A$ and $G_B$ is distinguishable via one-way LOCC with Alice going first. 
\end{enumerate}
\end{thm}

\begin{proof}
Let $\{a_k \kb{\varphi_k}{\varphi_k}\}_k$ be a POVM on Alice's system that is the first step in a one-way LOCC protocol to distinguish the elements of $\mathcal{S}$. Let $d = \dim\H_A$ and order the operators so that the first $d$ states in the list $\{ \ket{\varphi_k}\}_{k = 1}^d$ form a basis for $\H_A$. 

Let $V = \{1,2,\ldots n\}$ label the set of vertices of $G_A$. For each $k \le d$, the set of vertices $V_k := \{i: \bk{\varphi_k}{\psi_i^A} \ne 0\}$ corresponds to an independent set in $G_B$. Otherwise Bob would not be able to complete the measurement given outcome $k$. Since $\{\ket{\varphi_k}\}_{k = 1}^d$ spans $\H_A$, $\bigcup_{k = 1}^d V_k = V$. These two facts imply that the map $c(i) := \min\{ k: i \in V_k\}$ is a legal coloring of $G_B$ of size at most $d$. 

Since $d = \chi(G_B)$, $c$ is a minimal coloring and for any $k $, $\bigcup_{ j \ne k} V_j \ne V$. Hence, each $V_k$ contains a vertex that is unique to  $V_k$. Let us label Alice's states so that for each $k \le d$, $\ket{\psi_k^A}$ corresponds to this vertex.  Then for each $k$, $\ket{\psi_k^A}$ is the unique state in the orthogonal complement of $\{ \ket{\varphi_j}: j\le d,  j \ne k\}$. We can use this fact to show that Alice's POVM is in fact a projection onto an orthonormal basis. 

Suppose Alice's POVM contains more than $d$ operators, and choose some $\ket{\varphi_m}$ for $m > d$. We can write this in terms of our previous basis: 
\begin{equation*}
    \ket{\varphi_m} = \sum_{k = 1}^d \beta_k \ket{\varphi_k} .
\end{equation*}
For any $k$ with $\beta_k \ne 0$, the set of states $\{ \ket{\varphi_j}: j = m \mbox{ or } j\le d,  j \ne k\}$ forms a basis of $\H_A$. By repeating the argument above, we can find a vertex associated with $\ket{\varphi_m}$ so that $\ket{\psi_m^A}$ is the unique state in the orthogonal complement of $\{ \ket{\varphi_j}: j\le d,  j \ne k\}$, which means that $\vert \bk{\psi_m^A}{\psi_k^A}\vert = 1$.

Since $\{ \ket{\psi_j^A}\}_{j = 1}^d$ is linearly independent, there can only be one $k$ with $\vert \bk{\psi_m^A}{\psi_k^A}\vert = 1$, which implies that there is only one nonzero $\beta_k$. Thus $\kb{\varphi_m}{\varphi_m} = \kb{\varphi_k}{\varphi_k}$, and the inclusion of $\kb{\varphi_m}{\varphi_m}$ in the POVM is redundant.  This implies that Alice's POVM contains exactly $d$ rank one operators that add to the identity, which implies that they must be mutually orthogonal. 

Finally, we see that for each $k$, $\vert \bk{\varphi_j}{\psi_k^A}\vert  = \delta_{jk}$, which implies that the set $\{\ket{\psi_k^A}\}_{k = 1}^d$ is mutually orthogonal and corresponds to an independent set of size $d$ in $G_A$, hence $\alpha(G_A) \ge d$,  and by construction it corresponds to an independent set in $G_B$. Hence, $d \le \alpha(\overline{G_B}) \le \chi(G_B) = d$, implying equality. The equality of the chromatic number and the clique number of $G_B$ is the definition of a \emph{weakly perfect graph} in the literature.

To show that Alice's measurement is unique: For each $k$, we can look at the set of vertices $\{ i:i \notin \bigcup_{j\ne k} V_j\}$, which must form a clique in $G_A$. Since we have a set of $d$ mutually non-adjacent cliques, it must be that there is a unique state associated with each clique and that Alice's measurement must project onto it.

Finally, we see that for any two neighbors in $G_A$, $v_i\sim v_j$ implies that $\bk{\psi_i^A}{\psi_j^A} \ne 0$ and hence there exists at least one $k$ such that $\bk{\psi_i^A}{\varphi_k}\bk{\varphi_k}{\psi_j^A} \ne 0$, which implies that if $v_i \sim v_j$, there exists $k$ with $i, j \in V_k$. Thus, if we let $G$ be the union of the complete graphs on the $V_k$, then $G$ is a chordal graph with $G_A \le G \le \overline{G_B}$. By Theorem \ref{mainthm}, any representation of $G_A$ and $G_B$ must be distinguishable with one-way LOCC. 
\end{proof}

The theorem gives several conditions that are easy to identify and which preclude one-way LOCC discrimination. In particular, the following is an immediate consequence of conclusion (\ref{WeaklyPerfect}) in the theorem. 

\begin{cor}\label{AlphaLessThanChi} Let $\mathcal{S}$ be a collection of orthonormal product states in $\H_A \ot \H_B$, with associated graphs $G_A$ and $G_B$, such that $\dim\H_A = \chi(G_B)$. 

If $\alpha(\overline{G_B})$ is strictly less than $\chi(G_B)$, then the states cannot be distinguished via one-way LOCC with Alice going first. 

\end{cor}

A second corollary gives us a biconditional that follows from combining conclusion (\ref{ChordalCondition}) in Theorem \ref{ConverseTheorem} with Theorem \ref{mainthm}.

\begin{cor}\label{BiconditionalCorollary} Let $\mathcal{S}$ be a collection of orthonormal product states in $\H_A \ot \H_B$, with associated graphs $G_A$ and $G_B$, such that $\dim\H_A = \chi(G_B)$. 

The elements of $\mathcal{S}$ are distinguishable via one-way LOCC with Alice going first if and only if there exists a chordal graph $G$ with $G_A \le G \le \overline{G_B}$.
\end{cor}

\subsection{Examples} We present a number of examples motivated by the results of this section, with particular attention paid to the seminal work \cite{bennett1999quantum}. 

\begin{exa}\label{1steg}
    {\rm 
Recall that Example~\ref{Ex:4Cycle in 2D} gives a 4-cycle representation in dimension 2 for which $G_A=\overline{G_B}$ is not chordal, and, in that case, we have $\dim\H_A = \chi(\overline{G_B}) = 2$ and the states are indistinguishable by Corollary \ref{BiconditionalCorollary}. A more interesting graph-theoretical example is the nonlocality without entanglement states from Bennett, et al. \cite{bennett1999quantum}. For this well-known example of 9 states in $\C^3 \ot \C^3$, one can check that the graphs satisfy $G = G_A = \overline{G_B}$, where $G$ is the graph below with vertex states labeled corresponding to \cite{bennett1999quantum}: 

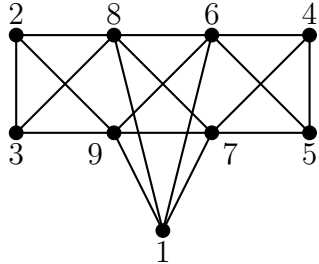
\begin{figure}[H]
\begin{center}
    
 \begin{tikzpicture}[scale=1.3, baseline=(current bounding box.north)]
\draw[thick] (0,0)--(3,0)--(3,1)--(0,1)--(0,0);
\draw[thick] (0,0)--(1,1)--(2,0)--(3,1);
\draw[thick] (0,1)--(1,0)--(2,1)--(3,0);
\draw[thick] (1,0)--(1.5,-1)--(2,0);
\draw[thick] (1,1)--(1.5,-1)--(2,1);
\draw[fill] (0,0) node[anchor=north] {$3$} circle [radius=2pt];
\draw[fill] (1,0) node[anchor=north east] {$9$} circle [radius=2pt];
\draw[fill] (2,0) node[anchor=north west] {$7$} circle [radius=2pt];
\draw[fill] (3,0) node[anchor=north] {$5$} circle [radius=2pt];
\draw[fill] (0,1) node[anchor=south] {$2$} circle [radius=2pt];
\draw[fill] (1,1) node[anchor=south] {$8$} circle [radius=2pt];
\draw[fill] (2,1) node[anchor=south] {$6$} circle [radius=2pt];
\draw[fill] (3,1) node[anchor=south] {$4$} circle [radius=2pt];
\draw[fill] (1.5,-1) node[anchor=north] {$1$} circle [radius=2pt];

\end{tikzpicture}

\end{center}

\end{figure}
This graph clearly has $\alpha(G) = 3$, which is also the dimension of Alice's system, so $\eta_+(G) = 3$ by Lemma~\ref{lemmainequality}. Since it has no simplicial vertices, we can apply Theorem \ref{NoSimplicial} and see immediately that it is not possible to distinguish this basis of $\C^3 \otimes \C^3$ with one-way LOCC. (We note the result in \cite{bennett1999quantum} is stronger, that it is impossible to distinguish the states with full LOCC operations.) 
}
\end{exa}

\begin{exa}\label{SubsetEx}
    {\rm Looking at subsets  of these states can only  make the distinguishability problem easier, and it is shown in \cite{bennett1999quantum} that removing any vertex (except ``9") from the states in Example \ref{1steg} results in coresponding states that are distinguishable with general LOCC measurements. Our argument shows that one-way LOCC is not sufficient to distinguish them. In fact, there exist subsets of 5 of these states that are also not distinguishable with one-way LOCC. Two such sets are given by the states that correspond to the graphs below, with $G_A = \overline{G_B}$ given for both. 

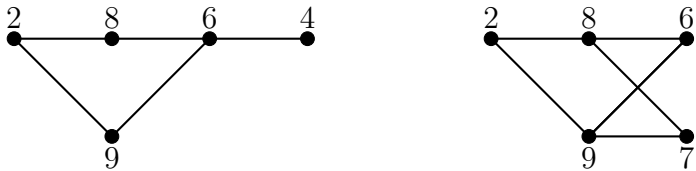
\begin{figure}[H]
\begin{minipage}{.5\textwidth}
\begin{center}
% \hfill
\begin{tikzpicture}[scale=1.3, baseline=(current bounding box.north)]
\draw[thick] (0,1)--(3,1);

\draw[thick] (0,1)--(1,0)--(2,1);
%\draw[fill] (0,0) circle [radius=2pt];
\draw[fill] (1,0) node[anchor=north] {$9$} circle [radius=2pt];
%\draw[fill] (2,0) circle [radius=2pt];
%\draw[fill] (3,0) circle [radius=2pt];
\draw[fill] (0,1) node[anchor=south] {$2$} circle [radius=2pt];
\draw[fill] (1,1) node[anchor=south] {$8$} circle [radius=2pt];
\draw[fill] (2,1) node[anchor=south] {$6$} circle [radius=2pt];
\draw[fill] (3,1) node[anchor=south] {$4$} circle [radius=2pt];
%\draw[fill] (1.5,-1) circle [radius=2pt];

 \end{tikzpicture} 
 
 %~ \hfill 
 \end{center}
 \end{minipage}%
\begin{minipage}{.4\textwidth}
 \begin{center}
 \begin{tikzpicture}[scale=1.3, baseline=(current bounding box.north)]
\draw[thick] (0,1)--(2,1);
\draw[thick] (2,1)--(1,0)--(2,0)--(1,1);
\draw[thick] (0,1)--(1,0)--(2,1);
%\draw[fill] (0,0) circle [radius=2pt];
\draw[fill] (1,0) node[anchor=north] {$9$} circle [radius=2pt];
\draw[fill] (2,0) node[anchor=north] {$7$} circle [radius=2pt];
%\draw[fill] (3,0) circle [radius=2pt];
\draw[fill] (0,1) node[anchor=south] {$2$} circle [radius=2pt];
\draw[fill] (1,1) node[anchor=south] {$8$} circle [radius=2pt];
\draw[fill] (2,1) node[anchor=south] {$6$} circle [radius=2pt];
%\draw[fill] (3,1) circle [radius=2pt];
%\draw[fill] (1.5,-1) circle [radius=2pt];

 \end{tikzpicture}~~~~~~~~~~~~~~~~~~~~~~~
 \end{center}
 \end{minipage}
\end{figure}

It is immediate in both cases that $\chi(G_B) = 3$. This, together with the fact that $\dim \H_A = 3$ and that neither of these graphs is chordal, allows us to apply Corollary~\ref{BiconditionalCorollary} and immediately conclude that neither of these sets can be distinguished with one-way LOCC. (On the other hand, in each case $G_B$ is chordal, so the states \emph{can} be distinguished if we alter our convention and allow Bob to measure first.)
}
\end{exa}

\begin{exa}\label{UPB} {\rm
The size of our set in Example \ref{SubsetEx} is potentially significant, as 5 is the minimum size of an unextendible product basis in $\C^3\ot \C^3$, and no unextendible product basis can be distinguished with separable measurements, which includes one-way LOCC. 

Following shortly on the work in \cite{bennett1999quantum}, there was work providing fundamental examples of unextendible product bases.  The authors in \cite{bennett1999unextendible,divincenzo2003unextendible} built examples of unextendible product bases and demonstrated the usefulness of representing product states with graphs. Their {\bf Tiles} is a set of 5 states in $\C^3 \ot \C^3$ corresponding to the states $\{2,4,6,8\}$ in Example \ref{SubsetEx} plus a stopper state that is adjacent to $2$ and $4$. The corresponding graph is a $5$-cycle, which cannot be distinguished with one-way LOCC by applying Corollary \ref{AlphaLessThanChi} and the logic in Lemma \ref{cycleLemma}. Unextendible product bases are not distinguishable with separable operations, so these LOCC results are not new, but they show how the various results may be applied in interesting cases.    }
\end{exa}
 
\begin{exa}
    {\rm 
The following is a generalization of Example \ref{1steg}, representing a set of $4d-3$ product states in $\mathbb C^d\ot \mathbb C^d$. For $d \ge 3$, consider the graph $B_d$ represented below. Each circle represents a complete graph or its complement, and heavy lines indicate that every vertex inside one circle is adjacent to every vertex inside the other. Example~\ref{1steg} is the case $d = 3$.

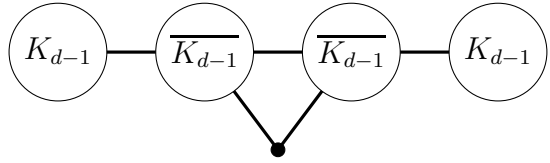
\begin{figure}[H]
\begin{center}
\begin{tikzpicture}[scale=1.3, baseline=(current bounding box.north)]
%\draw[very thick] (0,-.1)--(4.5,-.1);
%\draw[very thick] (0,.1)--(4.5,.1);
\draw[very thick] (0,0)--(4.5,0);
\draw[very thick]  (1.5,0)--(2.25,-1)--(3,0);
\foreach \x in {0,4.5}
   \draw[fill=white] (\x,0) circle [radius=0.5] node {$K_{d-1}$} ; 
\foreach \x in {1.5,3}
   \draw[fill=white] (\x,0) circle [radius=0.5] node {$\overline{K_{d-1}}$} ; 
\draw[fill] (2.25,-1) circle [radius=2pt] ; 

\end{tikzpicture}
\end{center}
\end{figure}

In general, this graph can be realized with the following states: Let $F$ be the $(d-1)\times (d-1)$ Fourier matrix and define $\tilde{F} = F \oplus 1$ to be the extension of $F$ to a $d$-dimensional space. Finally, let $X$ be a cyclic permutation on the standard basis in dimension $d$: $X = \sum_{k = 0}^{d-1} \kb{k+1}{k} \mod{d}$.
Then we can define the following states, where $k$ ranges from $0$ to $(d-2)$: 
\begin{align*}
    \ket{\psi_{1,k}} &= \tilde{F}\ket{0} \ot \ket{k} &\quad
    \ket{\psi_{3,k}} &= X\tilde{F}\ket{k} \ot \tilde{F}\ket{0}   \\
    \ket{\psi_{2,k}} &= \ket{k} \ot \ket{d-1} &\quad
    \ket{\psi_{4,k}} &=  \ket{d-1} \ot X\tilde{F}\ket{k}  \\
    \ket{\psi_0} &= \tilde{F}\ket{1} \ot\tilde{F}\ket{1} .
\end{align*}

It is straightforward to show that $\overline{B_d}$ is graph isomorphic to $B_d$ and that $\alpha(B_d) = \eta_+(B_d) = \chi(\overline{B_d}) = d$. This implies that our representation of $B_d$ in dimension $d$ is the smallest possible. Suppose Alice and Bob are working with a two-qudit system with a product state graph  representation $G_A = G_B = B_d$. Since $B_d$ contains no simplicial vertices, Theorem \ref{ConverseTheorem} implies that  one-way LOCC discrimination is not possible if Alice goes first; and the fact that the graph is self-complementary implies that it is not possible if Bob goes first, either.

It is also true that any subgraph  of $B_d$ that has $\alpha = d$ will continue to be indistinguishable with Alice going first, corresponding to this graph, with $ {a,b,c} \in [1,d-1]$:

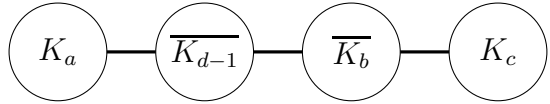
\begin{figure}[H]
\begin{center}
\begin{tikzpicture}[scale=1.3, baseline=(current bounding box.north)]
%\draw[very thick] (0,-.1)--(4.5,-.1);
%\draw[very thick] (0,.1)--(4.5,.1);
\draw[very thick] (0,0)--(4.5,0);
%\draw[very thick]  (1.5,0)--(2.25,-1)--(3,0);
\draw[fill=white] (0,0) circle [radius=0.5] node {$K_{a}$} ; 
   \draw[fill=white] (4.5,0) circle [radius=0.5] node {$K_{c}$} ; 

   \draw[fill=white] (1.5,0) circle [radius=0.5] node {$\overline{K_{d-1}}$} ; 
   \draw[fill=white] (3,0) circle [radius=0.5] node {$\overline{K_{b}}$} ; 
%\draw[fill] (0,0) circle [radius=2pt] ; 

\end{tikzpicture}
\end{center}
\end{figure}

Finally, for odd values of $d$, we can extend the original set to a complete basis of $\mathbb C^d\otimes \mathbb C^d$ by defining $G_1$ to be a single point $K_1$ and then recursively letting $G_d$ be given by:

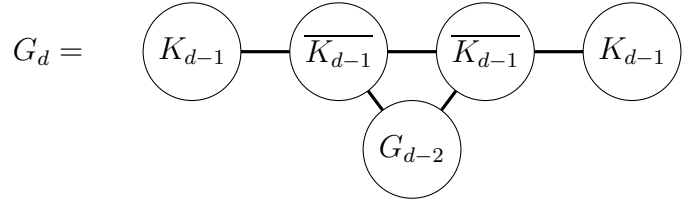
\begin{figure}[H]
\begin{center}
 \begin{tikzpicture}[scale=1.3, baseline=(current bounding box.north)]
%\draw[very thick] (0,-.1)--(4.5,-.1);
%\draw[very thick] (0,.1)--(4.5,.1);
\draw[very thick] (0,0)--(4.5,0);
\draw[very thick]  (1.5,0)--(2.25,-1)--(3,0);
\foreach \x in {0,4.5}
   \draw[fill=white] (\x,0) circle [radius=0.5] node {$K_{d-1}$} ; 
\foreach \x in {1.5,3}
   \draw[fill=white] (\x,0) circle [radius=0.5] node {$\overline{K_{d-1}}$} ; 
\node[left] at (-1,0) {$G_d=$};
\draw[fill=white] (2.25,-1) circle [radius=0.5] node {$G_{d-2}$} ; 

\end{tikzpicture}
\end{center}
\end{figure}

This is still self-complementary with minimum rank $d$ and $\alpha(G_d) = d$, so our theorems apply. 

We can also represent these states using a domino diagram, as in, e.g., \cite{bennett1999quantum, cohen2017general}. The diagram below is for $d = 5$, and gives motivation for calling this a Bullseye Basis. The numbers show how many states correspond to each rectangle: 

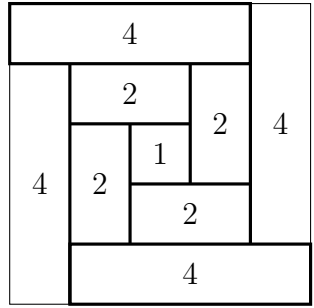
\begin{figure}[H]
\begin{center}
 \begin{tikzpicture}[scale=0.8, baseline=(current bounding box.north)]
\draw (0,0)--(5,0)--(5,5)--(0,5)--cycle; 
\draw[very thick]  (1,0)--(5,0)--(5,1)--(1,1)--cycle;
\node at (3,.5) {$4$};
\node at (3,1.5) {$2$};
\draw[very thick]  (0,4)--(4,4)--(4,5)--(0,5)--cycle;
\draw[very thick] (4,5)--(4,1);
\node at (2,4.5) {$4$};
\node at (2,3.5) {$2$};
\node at (0.5,2) {$4$}; \node at (1.5,2) {$2$};\node at (3.5,3) {$2$};\node at (4.5,3) {$4$};
\node at (2.5,2.5) {$1$};
\draw[very thick] (1,0)--(1,4);

\draw[very thick] (2,1)--(2,3);
\draw[very thick] (3,4)--(3,2);
\draw[very thick] (1,3)--(3,3);
\draw[very thick] (2,2)--(4,2);

\end{tikzpicture}
\end{center}
\end{figure}

}
\end{exa}

\begin{exa}
{\rm 
    One simple example for which $G_A \lneq \overline{G_B}$ has $G_B = C_5$, a 5-cycle, while $G_A = P_3 + P_2$, the union of an edge and a path of length 2. One can see that $G_A \le \overline{G_B}$ and $\chi(G_B) = 3$. 

    The representations of these graphs in $\C^3 \ot \C^3$ is fairly prescribed, for instance by the following 5 states: 
\begin{align*}
    \ket{\psi_{1}} &= \ket{0} \ot \ket{0} &\quad
    \ket{\psi_{2}} &= (\ket{0} + \ket{1}) \ot \ket{2}   \\
    \ket{\psi_{3}} &= \ket{1} \ot (\ket{0} + \ket{1}) &\quad
    \ket{\psi_{4}} &=  \ket{2} \ot(\ket{0} - \ket{1} + \ket{2})  \\
    &  &\quad \ket{\psi_5} &= \ket{2} \ot (\ket{1} + \ket{2}) .
\end{align*}
Theorem \ref{ConverseTheorem} states that any measurement to distinguish these sets starting with Alice must project onto an orthogonal basis, in this case the standard basis; and  Corollary \ref{BiconditionalCorollary} shows that this must be possible, since $G_A$ is chordal. 
}
\end{exa}

\section{Conclusions and Future Work}

We have built on our work from \cite{kribs2020vector}, extending the ways in which the areas of graph theory, orthogonal representations, and LOCC distinguishability can inform and enrich each other. In particular, in Section~\ref{section3} we established a characterization of chordal graphs in terms of one-way LOCC distinguishability, and we provided a number of applications and examples. In Section~\ref{LOCC and graph reps}, we reversed this, showing how graph parameters imply limitations on one-way LOCC measurements and providing graph-driven examples. 

Many open questions remain, including consideration of how other families of graphs might be characterized in terms of quantum operations; and how to algorithmically build one-way LOCC measurements from the corresponding graphs. It is also apparent that these ideas could be extended to two-way LOCC measurements and to sets of multipartite product states. This work is an invitation for researchers in disparate areas of mathematics to engage with these quantum-inspired questions. 

\strut

{\noindent}{\it Acknowledgements.} D.W.K. was partly supported by NSERC Discovery Grant 400160. R.P. was partly supported by NSERC Discovery Grant 400550.

\bibliographystyle{plain}

\bibliography{KMNPBibfile}

\end{document}